\documentclass[aps,
preprintnumbers,
nofootinbib,
twocolumn]{revtex4}
\usepackage{epsfig,graphics}
\usepackage{amsfonts,amssymb,amsmath} 
\usepackage{ifthen}
\usepackage{amsthm}
\usepackage{hyperref}
\usepackage{tikz}
\usetikzlibrary{arrows}

\newcommand{\comment}[1]{}

\newcommand{\ketbra}[2]{|#1\rangle\!\langle#2|}
\newcommand{\proj}[1]{\ketbra{#1}{#1}}

\newcommand{\tr}{{\rm tr}}

\newcommand{\cH}{\mathcal{H}}

\newcommand{\cM}{\mathcal{M}}
\newcommand{\cP}{\mathcal{P}}

\newcommand{\ot}{\otimes}

\newcommand{\cC}{{\rm C}}
\newcommand{\qQ}{{\rm Q}}

\newcommand{\nuparrow}{\hspace{-0.18ex}{\scriptscriptstyle{\not}}\hspace{0.18ex}\uparrow}
\theoremstyle{plain}
\newtheorem{theorem}{Theorem}
\newtheorem{lemma}{Lemma}

\theoremstyle{definition}

\begin{document}

\title{Inability of the entropy vector method to certify nonclassicality in linelike causal structures}

\date{\today}

\author{Mirjam \surname{Weilenmann}}
\email{msw518@york.ac.uk}
\affiliation{Department of Mathematics, University of York,
  Heslington, York, YO10 5DD, UK.}

\author{Roger \surname{Colbeck}}
\email{roger.colbeck@york.ac.uk}
\affiliation{Department of Mathematics, University of York,
  Heslington, York, YO10 5DD, UK.}

\begin{abstract}
  Bell's theorem shows that our intuitive understanding of causation
  must be overturned in light of quantum correlations. Nevertheless,
  quantum mechanics does not permit signalling and hence a notion of
  cause remains.  Understanding this notion is not only important at a
  fundamental level, but also for technological applications such as
  key distribution and randomness expansion.  It has recently been
  shown that a useful way to decide which classical causal structures
  could give rise to a given set of correlations is to use entropy
  vectors.  These are vectors whose components are the entropies of
  all subsets of the observed variables in the causal structure.  The
  entropy vector method employs causal relationships among the
  variables to restrict the set of possible entropy vectors.
  
  Here, we consider whether the same approach can lead to useful
  certificates of non-classicality within a given causal structure.
  Surprisingly, we find that for a family of causal structures that
  include the usual bipartite Bell structure they do not.  For all
  members of this family, no function of the entropies of the observed
  variables gives such a certificate, in spite of the existence of
  non-classical correlations.  It is therefore necessary to look
  beyond entropy vectors to understand cause from a quantum
  perspective.
\end{abstract}

\maketitle

\section{Introduction}
Correlation and causation are two different things.  They are however
connected.  Reichenbach's principle~\cite{Reichenbach1956} says that if
two events $X$ and $Y$ are correlated then either $X$ causes $Y$, $Y$
causes $X$ or they have a common cause.  In the standard view of
causation, having a common cause corresponds to the existence of a
shared random variable from which the observed correlations derive.
In other words, if $X$ and $Y$ have a common cause, then there exists
a random variable $C$ such that
$P_{XY}(x,y)=\sum_{c}P_{C}(c)P_{X|c}(x)P_{Y|c}(y)$.

In this work we will take a broader view of causation that allows the
common cause to be more general than a shared random variable, a direction that has also been considered in~\cite{Henson2014,Chaves2015,Pienaar2015,Fritz2015,Chaves2016}.  In
particular, we will allow shared quantum systems, so that if $X$ and
$Y$ have a \emph{quantum} common cause, then there exists a bipartite
quantum system $\rho$ and measurements described by POVMs $\{E_x\}_x$
and $\{F_y\}_y$ such that $P_{XY}(x,y)=\tr((E_x\ot F_y)\rho)$.

In this simple case, there is no separation between the sets of
classical and quantum correlations: for any $P_{XY}$ we can find a
classical common cause explanation as well as a quantum one.  However,
for more general causal structures this is not the case.
Bell~\cite{Bell1964} was the first to notice that quantum common
causes could allow for stronger correlations than their classical
counterparts.  Any restriction on the set of correlations that follows
under the assumption that any common causes are classical has been
termed a Bell inequality, and many such inequalities have been
discovered (e.g.~\cite{Clauser1969,GHZ}). The connection between Bell
inequalities and the literature on causal structures was elucidated
in~\cite{WoodSpekkens}, where a novel take on Bell's theorem was given.

Ruling out classical common causes is important in information theory,
and, especially, for device-independent
cryptography~\cite{Ekert1991,Mayers1998,Barrett2005b,Acin2006,Colbeck2009,
Colbeck2011,Pironio2010,Vazirani2014}. In
particular, it has recently been shown that the ability to demonstrate
non-classicality implies the ability to generate secure random
numbers~\cite{Miller2014}. It is therefore important to characterize
the set of classical correlations as far as possible.  Work in this
direction also helps us to understand the meaning of causation in
quantum theory.

In the standard Bell scenario, shown in Figure~\ref{fig:Bell}, the set
of classical correlations is well understood. However, as the scenario
is made more complicated, it rapidly becomes difficult to compute all
the Bell inequalities~\cite{WernerWeb}, and hence to precisely
separate the classical region from the non-classical.

\begin{figure}
\centering
\includegraphics[width=0.7\columnwidth]{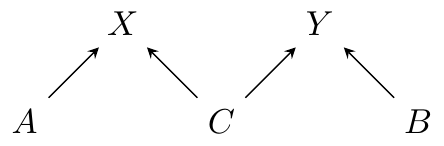}
\caption{The standard causal structure of a bipartite Bell experiment.
  Here $A$, $B$, $X$ and $Y$ are observed; $A$ and $B$ correspond to
  input settings and $X$ and $Y$ to outcomes.  If the common
  cause is classical, then the observed correlations satisfy
  $P_{ABXY}=\sum_CP_AP_BP_C P_{X|AC}P_{Y|BC}$.  In the case that $A$,
  $B$, $X$ and $Y$ are binary, the CHSH inequality~\cite{Clauser1969}
  can be derived. However, if the common cause is quantum this
  inequality can be violated, but Tsirelson's bound must hold
  instead~\cite{Tsirelson1993}.}
\label{fig:Bell}
\end{figure}

An approach to causal structures using entropy has recently been
developed~\cite{Steudel2015,Fritz2013,Chaves2012a,Chaves2013,Chaves2014,Chaves2014b,
  Chaves2015,Henson2014}.  The idea behind this approach is to study
the entropies of the observed variables that can be realised by
correlations within a given causal structure, rather than the
correlations themselves.  Note that entropy has been used in (at
least) two different ways in the causal structures literature.  In
this paper we study one of these ways and introduce the term
\emph{entropy vector method} for it.  When applied to $n$ observed
variables, the central object is the vector whose $2^n-1$ components
are the entropies of each subset of the variables (excluding the empty
set)\footnote{Note that when we refer to the entropy vector method, we
  consider entropies of observed variables. In particular, alternative
  approaches that condition on output values of some of the observed
  variables are not included in this terminology (see later in the
  introduction as well as in the discussion and in
  Appendix~\ref{sec:2} for details on alternative approaches of the
  latter kind).}.  This method is inviting because causal constraints
correspond to linear inequalities on entropies, rather than the
non-linear relations they imply for the probabilities.  This means
that entropy vectors are effective at distinguishing whether a set of
correlations can be generated within a particular causal structure.
Furthermore, the approach does not rely on any assumptions on the size
of the alphabet of the involved random variables.  In this paper we
study the use of this approach as a means of separating classical and
quantum versions of a given causal structure, focusing on a family of
``line-like'' causal structures that include the bipartite Bell
structure.

One of the advantages of the entropy vector method is its generality|it
applies to any causal structure.  However, other ways to use entropy
can be useful in this context and inequalities using entropy have been
derived for the bipartite Bell
scenario~\cite{Braunstein1988,Cerf1997}. These inequalities do not
concern the entropies of the observed variables directly, but rather
involve entropies of variables conditioned on particular outcomes of
other variables. This technique has recently been generalised to other
scenarios~\cite{Chaves2013,Fritz2013,Pienaar2016}. In the discussion
we elaborate on this alternative technique and we exemplify its
application to line-like causal structures in Appendix~\ref{sec:2}.
In contrast to the entropy vector method, this fine-grained technique
is not straightforwardly applicable to general causal structures and
for many causal structures it is not clear how to motivate entropic
inequalities of this type.

\section{The entropy vector method} \label{sec:method} We first
outline the classical case, initially introduced in~\cite{Yeung1997},
and its application to causal structures~\cite{Fritz2013}. For a
random variable, $X$, distributed according to $P_X$ we use the
Shannon entropy, $H(X):=-\sum_xP_X(x)\log P_X(x)$.\footnote{In this
  work lower case letters are used to denote particular instances of
  upper case random variables, and all random variables are taken to
  have finite alphabet.}  The conditional entropy is then defined by
$H(X|Y):=H(XY)-H(Y)$ and the conditional mutual information by
$I(X:Y|Z):=H(XZ)+H(YZ)-H(XYZ)-H(Z)$. A distribution over $n$ random
variables $X_1,\ldots,X_n$ has an associated entropy vector whose
$2^n-1$ components are the entropies of every subset of variables
(excluding the empty set).  Because they correspond to entropies of a
joint distribution, these components must satisfy certain constraints.
For example, they must be positive, obey monotonicity, i.e., $H(S)\leq
H(RS)$, and sub-modularity (or strong subadditivity), i.e.,
$H(RS)+H(ST)\geq H(RST)+H(S)$, where $R$, $S$, and $T$ denote disjoint
subsets of the $n$ random variables.  Monotonicity and sub-modularity
are equivalent to the positivity of the conditional entropy and
conditional mutual information respectively. This set of linear
constraints are called the \emph{Shannon constraints}.

Let ${\bf H}:P_{X_1\ldots X_n}\mapsto\mathbb{R}^{2^n-1}$ denote the
map from a joint distribution to its entropy vector.  We will consider
the set of entropy vectors that can be formed by applying ${\bf H}$ to
a probability distribution, i.e.,
$\Gamma^*_n=\{v\in\mathbb{R}^{2^n-1}:v={\bf H}(P_{X_1\ldots X_n})\}$
and its closure $\overline{\Gamma^*_n}$.  The latter is known to be
convex~\cite{Yeung1997}.  It is natural to ask whether any vector
$v\in\mathbb{R}^{2^n-1}$ that obeys the Shannon constraints is also in
$\overline{\Gamma^*_n}$. It turns out that this is the case for $n\leq
3$, but does not hold for larger $n$~\cite{Zhang1997}. Thus, the
Shannon constraints are necessary but not sufficient in order for a
vector to be the entropy vector of a probability distribution and the
set of vectors obeying these constraints is an outer approximation to
the set of achievable entropy vectors.

In order to account for the causal structure additional constraints
are included.  A \emph{causal structure} comprises a set of nodes
arranged in a directed acyclic graph (DAG).  A subset of these nodes
is designated as \emph{observed}.  If the causal structure is
\emph{classical}, each unobserved node has a corresponding random
variable.  For a causal structure $G$, we will use $G^\cC$ to denote
its classical version.  If all the nodes are observed, a probability
distribution is said to be \emph{compatible} with a classical causal
structure if it decomposes as
\begin{equation}\label{eq:compat}
P_{X_1\ldots X_n}=\prod_{i=1}^nP_{X_i|X_i^{\downarrow_1}}\, ,
\end{equation}
where $X_i^{\downarrow_1}$ denotes the parents of $X_i$ in the DAG.
For a classical causal structure $G^\cC$, we will use $\cP(G^\cC)$ to
denote the set of compatible distributions. If not all nodes are
observed, compatibility is defined by the existence of a joint
distribution that is compatible with the equivalent causal structure
with all nodes observed and having the correct marginal distribution
over the observed nodes (see Figure~\ref{fig:Bell} for an example). We
will denote this set $\cP_{\cM}(G^\cC)$.

A probability distribution decomposes as in~\eqref{eq:compat} if and
only if every variable $X_i$ is independent of its non-descendants
$X_i^{\nuparrow}$ conditioned on its parents $X_i^{\downarrow_1}$
(cf.\ Theorem~1.2.7 in~\cite{Pearl2009}).  Thus, for a DAG with $n$
variables, the compatibility constraints are implied by a minimal set
of (at most) $n$ equations.  In terms of entropies, these constraints
can be concisely written as
$I(X_i:X_i^{\nuparrow}|X_i^{\downarrow_1})=0$, which are linear
equalities in the entropies.

In general, the set of constraints on the underlying causal structure
implies additional constraints on the observed variables.  These can be
found by Fourier-Motzkin elimination~\cite{Williams1986} (see
also~\cite{Chaves2014,Chaves2015} for more details on its application
to causal structures).

For a causal structure $G$, we denote the set of achievable entropy
vectors by $\Gamma^*_\cM(G^\cC):=\{v:\exists P\in\cP_\cM(G^\cC)\text{
  with }v={\bf H}(P)\}$.  The closure of this,
$\overline{\Gamma^*_\cM}(G^\cC)$, is convex\footnote{This follows from
  the convexity of $\overline{\Gamma^*_n}$ and the fact that the causal constraints correspond to projections of this.}. 

\bigskip

The entropy vector approach was generalized to the quantum case
in~\cite{Pippenger2003}, and its application to causal structures
detailed in~\cite{Chaves2015}, which we now summarize.  The relevant
generalization of the Shannon entropy is the von Neumann entropy.  For
a system in state $\rho$ on $\cH_A$, it is defined by
$H(A):=-\tr(\rho\log\rho)$, and the quantum conditional entropy and
conditional mutual information are defined by replacing the Shannon
entropy by the von Neumann entropy in the classical definitions.  For
a quantum system comprising $n$ subsystems, we can again define a
vector $v\in\mathbb{R}^{2^n-1}$ whose entries are the corresponding
von Neumann entropies.  Like the Shannon entropy, the von Neumann
entropy is always positive and obeys sub-modularity.  However, it does
not in general obey monotonicity, but instead satisfies \emph{weak
  monotonicity}, i.e., $H(R)+H(S)\leq H(RT)+H(ST)$.  We call this set
of constraints \emph{von Neumann constraints}.  Like in the classical
case, these constraints are necessary, but not sufficient in order
that a given $v\in\mathbb{R}^{2^n-1}$ corresponds to the von Neumann
entropies of a joint quantum state~\cite{Pippenger2003}.

Rather than discuss the quantum version of arbitrary causal
structures, we consider here a restricted class that will be
sufficient for our purposes. In particular, we will consider causal
structures with only two generations, the first of which consists of
the unobserved variables and the second of the observed ones.  These
causal structures are, for example, relevant in the case that
spacelike separated observations are made (so that none of the
observed variables can be the cause of any other).  For convenience,
we will use $C_i$ or $C$, $D$, $E$ etc.\ for unobserved nodes, and
$X_i$ or $W$, $X$, $Y$ etc.\ for observed ones.  In this case, if the
causal structure is \emph{quantum}, each edge of the graph has an
associated Hilbert space, which can be labelled by the parent and
child, e.g., there will be a Hilbert space $\cH_{C_X}$ if the DAG
contains $C\rightarrow X$.  For each unobserved node there is an
associated quantum state, a density operator on the tensor product of
the Hilbert spaces associated with the edges coming from that node.
For each observed node there is an associated POVM that acts on the
tensor product of the Hilbert spaces associated with the edges that
meet at that node. The corresponding correlations are those resulting
from performing the specified POVMs on the relevant systems via the
Born rule.  An example is shown in Figure~\ref{fig:quantum}.  With
respect to a causal structure $G$, we use $\cP_{\cM}(G^{\qQ})$ to
denote the set of distributions on the observed nodes that can be
realised if the causal structure is quantum.

\begin{figure}
\centering
\includegraphics[width=0.65\columnwidth]{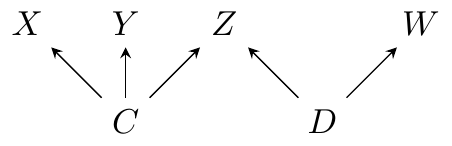}
\caption{Example causal structure.  If this is quantum, then the
  realisable correlations are those formed by measuring states of the
  form $\rho_{C_XC_YC_Z}\ot\sigma_{D_ZD_W}$ with separate measurements
  on $\cH_{C_X}$, $\cH_{C_Y}$, $\cH_{C_ZD_Z}$ and $\cH_{D_W}$.}
\label{fig:quantum}
\end{figure}

In the entropic picture, there is an entropy for each observed node
and for each edge of the DAG in question (for convenience we will
refer to both of these as subsystems in the following\footnote{Note,
  however, that they are not all subsystems of one joint quantum
  state.}).  While for $n$ jointly distributed random variables, all
the joint entropies make sense, this is no longer the case in a
quantum causal structure with $n$ subsystems. In particular, the
subsystems corresponding to the edges that meet at an observed node do
not coexist with the outcome at that node and hence there is no joint
quantum state from which the joint entropy can be derived.  For
example, if a measurement is performed on $\cH_{C_X}\ot\cH_{D_X}$ with
outcome $X$, then $H(C_XD_XX)$ is not well-defined, although
$H(C_XD_X)$ is\footnote{Note also that in the classical case the
  analogous argument fails as information can always be copied.}.  To
avoid this problem, the approach only considers entropies of
coexisting sets. Two subsystems are said to \emph{coexist} if neither
is a quantum ancestor of the other, and a set of subsystems that
pairwise coexist form a \emph{coexisting set}.

Within each coexisting set the von Neumann constraints hold.  However,
since the observed subsystems are classical, some of the weak
monotonicity constraints can be replaced by monotonicity.  For
example, if either $R$ or $S$ is a set of classical variables, then the monotonicity
constraint $H(RS)\geq H(R)$ holds.

The causal constraints are accounted for by the condition that two
subsets of a coexisting set are independent (and hence have zero
mutual information between them) if they have no shared ancestors.  To
connect different coexisting sets, data processing inequalities are
used.  For example, if a measurement is performed on
$\cH_{C_Y}\ot\cH_{D_Y}$ with outcome $Y$, then $I(C_YD_Y:X)\geq
I(Y:X)$ (cf.\ Figure~\ref{fig:quantum}).

Like in the classical case, we denote the set of achievable entropy
vectors by $\Gamma^*_\cM(G^\qQ):=\{v:\exists P\in\cP_\cM(G^\qQ)\text{
  with }v={\bf H}(P)\}$, and its closure
$\overline{\Gamma^*_\cM}(G^\qQ)$ is again convex.

\section{Line-like causal structures} \label{sec:pn}
For the remainder of this paper, we consider the family of line-like
causal structures shown in Figure~\ref{fig:Pn}. The causal structure
$P_n$ has observed nodes $X_1,X_2,\ldots,X_n$.  Each pair of
consecutive observed nodes $X_i$ and $X_{i+1}$ has an unobserved
parent $C_i$.

\begin{figure}
\centering
\includegraphics[width=0.98\columnwidth]{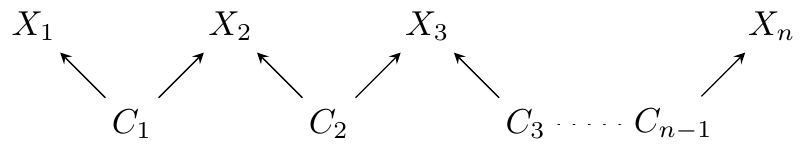}
\caption{The causal structure $P_n$. The nodes $X_i$
represent observed variables, whereas the 
$C_i$
denote the unobserved classical or quantum systems.}
\label{fig:Pn}
\end{figure}

The case $n=4$ is in one-to-one correspondence with the bipartite Bell
causal structure of Figure~\ref{fig:Bell}~\cite{Fritz2012}.  To make
the identification, take $X_1=A$, $X_2=X$, $X_3=Y$, $X_4=B$ and
$C_2=C$.  We can assume without loss of generality that $C_1=A$ and
$C_3=B$: the same set of observed correlations can be generated in
either case\footnote{To see this, note that for any bipartite quantum
  state $\rho_{CD}$, if $C$ is measured to generate $A$, then the post
  measurement state has the form $\sum_aP_A(a)\proj{a}\ot\rho^a_D$.
  The same joint state can be generated by sharing $A=a$ with
  distribution $P_A(a)$ and simulating the statistics of the state
  $\rho^a_D$ at $D$ conditioned on $A=a$.}.  

In the classical case the
node $C$ corresponds to a local hidden variable.  Free choice of
settings, crucial to the derivation of a Bell inequality, is naturally
encoded in the causal structure (e.g., $P_{A|BYC}=P_A$ follows as $A$
has no parents but $BYC$ as its non-descendants), as are the
conditions of \emph{local causality}, that
$P_{XY|ABC}=P_{X|AC}P_{Y|BC}$.  The only difference between $P_4^\cC$
and the quantum case, $P_4^\qQ$, is the nature of the node $C$.
Bell's original argument then implies that there are non-classical
correlations, i.e., there are distributions in $\cP_{\cM}(P_4^{\qQ})$
that are not in $\cP_{\cM}(P_4^{\cC})$.

In the following we will prove that, in spite of this separation, by
looking at the entropy vectors no distinction can be made.  This is stated more formally as follows.
\begin{theorem}\label{thm:1}
  $\overline{\Gamma^*_\cM}(P_n^\qQ)=\overline{\Gamma^*_\cM}(P_n^\cC)$ for all $n\in\mathbb{N}$.
\end{theorem}

Note that for $n\leq 3$, $\cP_\cM(P^\cC_n)=\cP_\cM(P^\qQ_n)$ and hence
in these cases the lemma immediately follows~\cite{Fritz2012}.  We
proceed to give the argument for $n=4$, deferring the general case to Appendix~\ref{sec:1}. 

Note also that the $n=5$ case is closely related to the so-called
bilocality scenario~\cite{Branciard2010,Branciard2012}, introduced in
the context of entanglement swapping. The difference to $P_5$ is that
bilocality also allows an additional observed ``input'' to the central
node. In fact, following an analogous argument to that of
Theorem~\ref{thm:1} reveals that in the bilocality scenario there is
also no separation between the classical and quantum entropy cones.

\begin{proof}[Proof of Theorem~\ref{thm:1} for $n=4$]
  The entropy vector of the joint distribution of $A$, $X$, $Y$ and
  $B$ has to obey the Shannon inequalities in both $P_4^\cC$ and
  $P_4^\qQ$. In addition, the causal structure directly implies the
  following independences among the four observed
  variables\footnote{Note that $A$ and $YB$ do not share any ancestors
    (similarly $AX$ and $B$).}:
\begin{equation}\label{eq:P4indep}
\begin{split}
I(A:YB)&=0,\\
I(AX:B)&=0.
\end{split}
\end{equation}
In both the classical and the quantum case, if the unobserved
subsystems are included, there are further valid (in)equalities
implied by the causal structure. The following argument shows,
however, that these do not impart any additional constraints on the
entropy vector of the observed nodes: in $P_4$ the Shannon
inequalities together with~\eqref{eq:P4indep} fully characterise the
set of achievable entropy vectors of the observed nodes in both the
classical and quantum case.

The Shannon inequalities on four variables together
with~\eqref{eq:P4indep} are necessary conditions on a vector
$v\in\mathbb{R}^{15}$ in order that there is a distribution $P_{AXYB}$
in $\cP_{\cM}(P^\cC_4)$ with ${\bf H}(P_{AXYB})=v$.  They therefore
form an an outer approximation to $\overline{\Gamma^*_\cM}(P^\cC_4)$.
This outer approximation is a convex cone that can equivalently be
expressed via its extremal rays.  Conversion between these two
descriptions can be conveniently done using software such as
PORTA~\cite{Christof} or PANDA~\cite{Loerwald} and results in the
following rays, where the components are ordered as
\begin{align*}
(&H(A),H(X), H(Y), H(B), H(AX), H(AY), \\ &H(AZ), H(XY), 
H(XB), H(YB), H(AXY), \\
&H(AXB), H(AYB), H(XYB), H(AXYB)),
\end{align*}
\begin{alignat*}{3}
&\text{(i)} \quad & &1 1 1 1 2 2 2 2 2 2 3 3 3 3 3 \\
&\text{(ii)} \quad & &0 1 1 1 1 1 1 2 2 2 2 2 2 2 2 \\
&\text{(iii)} \quad & &1 1 1 0 2 2 1 2 1 1 2 2 2 2 2 \\
&\text{(iv)} \quad & &0 0 0 1 0 0 1 0 1 1 0 1 1 1 1 \\
&\text{(v)} \quad & &0 0 1 0 0 1 0 1 0 1 1 0 1 1 1 \\
&\text{(vi)} \quad & &0 1 0 0 1 0 0 1 1 0 1 1 0 1 1 \\
&\text{(vii)} \quad & &1 0 0 0 1 1 1 0 0 0 1 1 1 0 1 \\
&\text{(viii)} \quad & &0 0 1 1 0 1 1 1 1 1 1 1 1 1 1 \\
&\text{(ix)} \quad & &0 1 1 0 1 1 0 1 1 1 1 1 1 1 1 \\
&\text{(x)} \quad & &1 1 0 0 1 1 1 1 1 0 1 1 1 1 1.
\end{alignat*}

If each of these rays is achievable using a distribution in
$\cP_\cM(P_4^\cC)$ then, by convexity of
$\overline{\Gamma^*_\cM}(P_4^\cC)$, the outer approximation is tight.  In other
words, any vector $v$ that obeys the Shannon constraints
and~\eqref{eq:P4indep} is achievable, i.e., in
$\overline{\Gamma^*_\cM}(P_4^\cC)$.  We establish this by taking
$C_1$, $C_2$ and $C_3$ to be uniform random bits and use the following
functions:
\begin{itemize}
\item (i): Take $A=C_1$, $X=C_1\oplus C_2$, $Y=C_2 \oplus C_3$ and
  $B=C_3$.
\item (ii): Let $A=1$ be deterministic and choose $X=C_2$, $Y=C_2\oplus
  C_3$ and $B=C_3$. (iii) can be achieved with an analogous strategy, where
  $B=1$ is the deterministic variable.
\item (iv): Choose $A=X=Y=1$ and $B=C_3$. (v), (vi) and (vii) are permutations of this strategy.
\item (viii): Let $A=X=1$ be deterministic and let $Y=B=C_2$. (ix) and (x) are permutations of this.
\end{itemize}

The outer approximation of the set of entropy vectors that are
achievable classically, $\overline{\Gamma^*_\cM}(P_n^\cC)$, given here
is also an outer approximation to $\overline{\Gamma^*_\cM}(P_n^\qQ)$.
Since the extremal rays are achievable the lemma follows.
\end{proof}

\section{Discussion}
Although for all $n\geq 4$ there are distributions in
$\cP_\cM(P_n^\qQ)$ that cannot be achieved in $\cP_\cM(P_n^\cC)$ the
entropy vector approach we have outlined is unable to detect this.
Even correlations that in other contexts are thought of as strongly
non-classical have this masked under the mapping to entropy vectors:
no function of the entropy vector acts as a certificate of
non-classicality in these causal structures. It is an interesting open
question as to whether this is generic: i.e., can entropy vectors ever
detect the difference between classical and quantum versions of a
given causal structure?  We discuss this question in more detail
in~\cite{non_shan}.

Because of the shortcomings of the entropy vector method, other
techniques will be needed to separate classical and quantum causal
structures.  Recently, other approaches to this have been developed,
one involving polynomial Bell
inequalities~\cite{Chaves2015a,Rosset2016} and the other drawing on
tools from algebraic geometry~\cite{Lee2015}.

As mentioned in the introduction, for certain causal structures
(including line-like ones), an alternative entropic technique can be
applied, as first introduced by Braunstein and
Caves~\cite{Braunstein1988}. In our terminology, the inequality
of~\cite{Braunstein1988} states that in the causal structure
$P_4^{\cC}$
\begin{equation}\label{eq:bc}
H(Y|X)_{11}+H(X|Y)_{10}+H(X|Y)_{01}-H(X|Y)_{00} \geq 0\, , 
\end{equation}
where $H(X|Y)_{ab}$ is the conditional entropy of the
conditional distribution $P_{XY|A=a,B=b}$.

The crucial idea behind the derivation of inequality~\eqref{eq:bc} is
that in the classical case there exists a joint distribution
$P'_{X_0X_1Y_0Y_1}$ whose marginals satisfy
$P'_{X_aY_b}=P_{XY|A=a,B=b}$ for all $a$ and
$b$~\cite{Fine1982,Fine1982a}. In the quantum case there is no such
distribution in general, and hence~\eqref{eq:bc} does not apply. Such
inequalities are not obtained with the entropy vector method because
the latter does not consider conditioning on particular outcomes.

It was shown in~\cite{Chaves2013} that every non-local
distribution in $P_4$ can be used to violate such an inequality if one
takes an appropriate convex combination with a local
distribution.
In fact, the inequality~\eqref{eq:bc} and its permutations are the only relevant inequalities for two measurements with dichotomic outcomes for each party~\cite{Fritz2013}.
These inequalities can also be generalized to the chained Bell
inequalities~\cite{Braunstein1988}, which allow for $A$ and $B$ to take any number of values~\cite{Chaves2013,Fritz2013}.

Entropic inequalities of this type (i.e., after conditioning on output values of some of the observed variables) may arise in other classical causal structures~\footnote{Note, however, that the application of this method to general causal structures is not so straightforward, as the justification of a statement similar to Fine's theorem is for many of them not evident.}. In Appendix~\ref{sec:2}, we show how additional entropic inequalities for $P_5^{\cC}$ and $P_6^{\cC}$ may be derived with this technique. It is an open question, however, as to whether any quantum violations of these extra inequalities exist.

It is natural to ask whether the entropy vector method can be used with other
entropy measures, the family of R\'enyi
entropies~\cite{Renyi1960_MeasOfEntrAndInf} being a natural
alternative, as considered in~\cite{Linden2013a}.  These do not obey
sub-modularity and hence the set of allowed entropy vectors is (using
known inequalities) far less constrained than in the von Neumann case.
Although R\'enyi conditional entropies satisfy $H_{\alpha}(A|BC)\leq
H_{\alpha}(A|B)$~\cite{Petz,TCR,MDSFT,FL,Beigi}, because the
conditional R\'enyi entropy cannot be expressed as a difference of
unconditional entropies, these relations do not lead to constraints on
the R\'enyi entropy vector.  Including conditional entropies as
separate elements of the entropy vector would allow use of these
relations, but given the expanded length of the vector and the
comparatively small number of additional constraints, we don't expect
this to be fruitful without further inequalities between R\'enyi
entropies.

One can also look at causal structures that allow post-quantum
non-signalling systems, such as non-local
boxes~\cite{Tsirelson1993,Popescu1994FP}, to be shared.  One approach
to this has been presented in~\cite{Henson2014}. In the case of $P_4$,
this yields the constraints of~\eqref{eq:P4indep} on the observed
variables. Hence, the proof of Theorem~\ref{thm:1} can be used to show
that functions of the entropy vector of the observed variables cannot
detect post-quantum non-locality either. Whether the entropy vectors
are ever able to encode information about the physical nature of the
involved variables, rather than mere independences, remains an open
question.

\bigskip

\begin{acknowledgements}
We thank Rafael Chaves for comments on a previous version of this work. RC is partly supported by the EPSRC's Quantum Communications Hub. 
\end{acknowledgements}

\appendix 

\section{Proof of Theorem~1} \label{sec:1} 
We rely on the following lemma~\cite{Yeung1997}
\begin{lemma}
  Consider $n$ variables $X_1,X_2,\ldots,X_n$ and define
  $\Omega:=\left\{X_1,X_2,\ldots,X_n\right\}$.  Taking positivity of
  each entropy to be implicit, the Shannon inequalities for these are
  generated from a minimal set of $n+n(n-1)2^{n-3}$ inequalities:
\begin{eqnarray}
H(\Omega | \Omega \setminus \{X_i\} )&\geq&0,\label{eq:monot}\\
I(X_i:X_j|X_S)&\geq&0,\label{eq:submod}
\end{eqnarray}
where the first is needed for all $X_i\in\Omega$ and the second for
all $X_S\subsetneq\Omega$, $X_i,X_j\in\Omega$, $X_i,X_j\notin X_S$,
$i<j$.
\end{lemma}

Now take $X_1,X_2,\ldots,X_n$ to be the observed nodes in $P_n$ and
for $i+1<j$ define $M_{i,j}:=\left\{ X_k \right\}_{k=i+1}^{j-1}$ as
the set of nodes between $X_i$ and $X_j$. The first part of the proof
of Theorem~1 is to show that from these $n+n(n-1)2^{n-3}$
Shannon inequalities at most $\frac{n(n+1)}{2}$ are not implied by the
conditional independence constraints and the remaining Shannon
inequalities.

To directly read conditional independences off the DAG, we use a
condition known as d-separation. For a classical causal structure, if
$X$, $Y$ and $Z$ are disjoint sets of variables then $X$ and $Y$ are
said to be \emph{d-separated} by $Z$ if every path from a node in $X$
to a node in $Y$ contains one of (i) $c\rightarrow z\rightarrow d$
with $z\in Z$, (ii) $c\leftarrow z\rightarrow d$ with $z\in Z$ or
(iii) $c\rightarrow e\leftarrow d$ with $e\notin Z$.  As shown by
Verma and Pearl~\cite{Verma1990}, if a distribution is compatible with
a classical causal structure in which $X$ and $Y$ are d-separated by
$Z$, then $I(X:Y|Z)=0$.

\begin{lemma}\label{lemma:middlenodes}
Within the causal structure $P^\cC_n$, all of the sub-modularity inequalities~\eqref{eq:submod} with $M_{i,j}\not\subseteq X_S$
are implied by the causal constraints.
\end{lemma}
\begin{proof}
Let $M_{i,j} \not\subseteq X_S$, then there is at least one node $X_{k} \not\in X_S$ with $i<k<j$. For each such node we can partition $X_S=\left\{ X_S^{k-},~X_S^{k+} \right\}$, where $X_S^{k-}$ contains all $X_l \in X_S$ with $l<k$ and $X_S^{k+}$ contains the elements with $X_l \in X_S$ with $l>k$ (note that both sets might be empty). Since $\{X_i\}\cup X_S^{k-}$ is d-separated
from $\{X_j\}\cup X_S^{k+}$ we have
\begin{eqnarray*}
H(\left\{X_i,~X_j \right\} \cup X_S)&=&H(\left\{X_i \right\} \cup X_S^{k-})+H(\left\{X_j \right\} \cup X_S^{k+}),\\
H(\left\{X_i \right\} \cup X_S)&=&H(\left\{X_i \right\} \cup X_S^{k-})+H(X_S^{k+}),\\
H(\left\{X_j \right\} \cup X_S)&=&H( X_S^{k-})+H(\left\{X_j \right\} \cup X_S^{k+}), \\
H(X_S)&=&H( X_S^{k-})+H(X_S^{k+}),
\end{eqnarray*}
and thus~\eqref{eq:submod} is obeyed with equality.
\end{proof}

\begin{lemma}
  Within the causal structure $P_n^\cC$, the $\frac{n(n-1)}{2}$
  sub-modularity constraints of the form $I(X_i:X_j|M_{i,j})\geq 0$ for
  all $X_i,X_j$ with $i<j$ imply all sub-modularity
  constraints~\eqref{eq:submod}.
\end{lemma}
\begin{proof}
  Lemma~\ref{lemma:middlenodes} shows this to hold in the
  case $M_{i,j}\not\subseteq X_S$.  Thus, we restrict to the case
  $M_{i,j}\subseteq X_S$. Let us write $X_S=M_{i,j}\cup X_T$, where
  $X_T=X_S\setminus M_{i,j}$.

  First consider the case where $X_{i-1},X_{j+1}\notin X_T$. Here
  $M_{i-1,j+1}$ and $X_T$ are d-separated and hence
  \begin{eqnarray*}
    H(\{X_i,X_j\}\cup M_{i,j}\cup X_T)&=&H(\{X_i,X_j\}\cup M_{i,j})+H(X_T)\\
    H(\{X_i\}\cup X_T)&=&H(X_i)+H(X_T)\\H(\{X_j\}\cup X_T)&=&H(X_j)+H(X_T)\\H(M_{i,j}\cup X_T)&=&H(M_{i,j})+H(X_T)
  \end{eqnarray*}
  so that $I(X_i:X_j|M_{i,j}\cup X_T)=I(X_i:X_j|M_{i,j})$.

Next, consider the case where $X_k\in X_T$ for $k=j+1,j+2,\ldots,j+L$, but $X_{i-1},X_{j+L+1}\notin X_T$. By
d-separation, we have $I(X_i:X_j|M_{i,j}\cup
X_T)=I(X_i:X_j|M_{i,j}\cup\{X_{j+1},\ldots,X_{j+L}\})$,
and the latter expression can be more concisely written as
$I(X_i:X_j|M_{i,j}\cup M_{j,j+L+1})$.  Then,
\begin{align*}
&I(X_i:X_j|M_{i,j}\cup
M_{j,j+L+1})\\
&=I(X_i:M_{j,j+L+1}\cup\{X_j\}|M_{i,j})-I(X_i:M_{j,j+L+1}|M_{i,j})\\
&=I(X_i:M_{j,j+L+1}\cup\{X_j\}|M_{i,j})\\
&=I(X_i:X_j|M_{i,j})+I(X_i:M_{j,j+L+1}|M_{i,j+1})\\
&=I(X_i:X_j|M_{i,j})+I(X_i:M_{j+1,j+L+1}\cup\{X_{j+1}\}|M_{i,j+1}),
\end{align*}
where we have used $I(X_i:M_{j,j+L+1}|M_{i,j})=0$, which follows from
d-separation.  Noting the relation between the last term in the final line and the
third line, we can proceed to recursively decompose the
expression into
\begin{align}
I(X_i:X_j|&M_{i,j}\cup M_{j,j+L+1})\nonumber\\
&=\sum_{l=0}^L I(X_i:X_{j+l}|M_{i,j+l})\, .
\label{eq:removej}
\end{align}

Now suppose $X_k\in X_T$ for $k=i-1,i-2,\ldots i-K$ and
$k=j+1,j+2,\ldots,j+L$, but $X_{i-K-1},X_{j+L+1}\notin X_T$. By d-separation, we have $I(X_i:X_j|M_{i,j}\cup
X_T)=I(X_i:X_j|M_{i,j}\cup\{X_{i-K},\ldots,X_{i-1}\}\cup\{X_{j+1},\ldots,X_{j+L}\})$, and the latter expression can be more concisely written as $I(X_i:X_j|M_{i,j}\cup M_{i-K-1,i}\cup M_{j,j+L+1})$. Then,
\begin{align*}
&I(X_i:X_j|M_{i,j}\cup
M_{i-K-1,i}\cup
M_{j,j+L+1})\\
=&I(M_{i-K-1,i}\cup\{X_i\}:X_j|M_{i,j}\cup
M_{j,j+L+1})\\
&-I(M_{i-K-1,i}:X_j|M_{i,j}\cup
M_{j,j+L+1})\\
=&I(M_{i-K-1,i}\cup\{X_i\}:X_j|M_{i,j}\cup
M_{j,j+L+1})\\
=&I(X_i:X_j|M_{i,j}\cup
M_{j,j+L+1})\\
&+I(M_{i-K-1,i}:X_j|M_{i-1,j}\cup
M_{j,j+L+1})\\
=&I(X_i:X_j|M_{i,j}\cup
M_{j,j+L+1})\\
&+I(M_{i-K-1,i-1}\cup\{X_{i-1}\}:X_j|M_{i-1,j}\cup
M_{j,j+L+1}),
\end{align*}
where we have used $I(M_{i-K-1,i}:X_j|M_{i,j}\cup
M_{j,j+L+1})=0$, which follows from d-separation.
Noting the relation between the last term in the final line and the
third line, we can hence proceed to recursively decompose the
expression into
\begin{align*}
I(X_i:X_j|&M_{i,j}\cup
M_{i-K-1,i}\cup
M_{j,j+L+1})\\
&=\sum_{k=0}^K I(X_{i-k}:X_j|M_{i-k,j}\cup
M_{j,j+L+1})\, .
\end{align*}
The latter can then be decomposed using~\eqref{eq:removej}.
\end{proof}

Including the $n$ monotonicity constraints, there are at most
$\frac{n(n+1)}{2}$ Shannon inequalities that are not implied by the
conditional independence relations of $P_n^\cC$. These inequalities
constrain a pointed polyhedral cone with the zero vector as its
vertex. They hold for all entropy vectors in $P_n^\cC$ and thus
approximate the entropy cone $\overline{\Gamma^*_\cM}(P_n^\cC)$ from the outside. They are
also valid for $\overline{\Gamma^*_\cM}(P_n^\qQ)$ (recall that two subsets of a coexisting
set are independent if they have no shared ancestors).  Note that the
causal constraints reduce the effective dimensionality of the problem
to $\frac{n(n+1)}{2}$, since the entropies of contiguous sequences are
sufficient to determine all entropies\footnote{There are $n$
  contiguous sequences of length $1$, $\{H(X_i)\}_{i=1}^n$, $n-1$ of
  length $2$, $\{H(X_iX_{i+1})\}_{i=1}^{n-1}$, and so on, leading to
  $\sum_{i=1}^ni=\frac{n(n+1)}{2}$ in total.}.

The $\frac{n(n+1)}{2}$ inequalities can lead to at most
$\frac{n(n+1)}{2}$ extremal rays, which corresponds to the number of
ways of choosing $\frac{n(n+1)}{2}-1$ inequalities to be
simultaneously obeyed with equality. In the following we show that
this bound is saturated by constructing $\frac{n(n+1)}{2}$ entropy
vectors from probability distributions in $P_n^\cC$, each of which lies
on a different extremal ray.

Consider the following set of distributions in $P_n^\cC$ (leading to
corresponding entropy vectors). Let $\{C_i\}_{i=1}^{n-1}$ be uniform
random bits, and $1\leq i\leq j\leq n$.  For each $i,~j$ we define a distribution $D_{i,j}$.
\begin{itemize}
\item For $i\leq n-1$, $D_{i,i}$ is formed by taking
  $X_i=C_i$ and $X_k=1$ for all $k\neq i$, while $D_{n,n}$ has
  $X_i=C_{i-1}$ and $X_k=1$ for all $k\neq i$.
\item For $i<j$, $D_{i,j}$ is constructed in the following. Note that depending on $i$ and $j$, each of the parts indexed by $k$ below may also be empty.
\begin{itemize}
\item $X_k=1$ for $1\leq k\leq i-1$,
\item $X_i=C_i$,
\item $X_k=C_{k-1} \oplus C_k$ for $i+1\leq k\leq j-1$, where $\oplus$ denotes addition modulo 2,
\item $X_j=C_{j-1}$,
\item $X_k=1$ for $j+1\leq k\leq n$.
\end{itemize}
\end{itemize}
Note that the set of distributions $\{D_{i,j}\}_{i,j}$ for $1 \leq i \leq j \leq n$ is in one-to-one correspondence with the contiguous sequences from $\Omega$.

\begin{lemma}\label{lem:5}
The $\frac{n(n+1)}{2}$ entropy vectors of the probability distributions $\{D_{i,j}\}_{i,j}$ with $1 \leq i \leq j \leq n$ are extremal rays of $\overline{\Gamma^*_\cM}(P_n^\cC)$.
\end{lemma}

\begin{proof}
It is sufficient to prove the following:
\begin{itemize}
\item For each $i$, $D_{i,i}$ obeys all of the Shannon equalities with equality except the monotonicity relation $H(\Omega)-H(\Omega\setminus\{X_i\})\geq 0$, which is a strict inequality.
\item For $i<j$, $D_{i,j}$ obeys all of the Shannon inequalities with equality except $I(X_i:X_j|M_{i,j})\geq 0$, which is a strict inequality.
\end{itemize}

For the $n$ distributions $D_{i,i}$ all variables are independent and thus their entropy vectors automatically satisfy all sub-modularity inequalities with equality. Furthermore, for any $X_S \subsetneq \Omega$
with $X_i\notin X_S$ we have $H(\{X_i\}\cup X_S)=H(X_i)$. Thus, for $j\neq i$ we have
\begin{equation*}
H(\Omega)-H(\Omega\setminus\{X_j\})=0\, ,
\end{equation*}
while for $j=i$
\begin{equation*}
\begin{split}
H(\Omega)-H(\Omega\setminus\{X_j\})&=H(X_i) \\
&>0.
\end{split}
\end{equation*} 
This establishes the first statement.

Consider now the $(n-1)!$ distributions $D_{i,j}$ with $i<j$.  We
first deal with the monotonicity constraints.  For $k<i$ and $k>j$, we
have $H(\Omega)=H(\Omega\setminus X_k)=j-i$. Similarly, since any
$j-i-1$ elements of $M_{i-1,j+1}$ are sufficient to determine the
remaining element, we also have $H(\Omega\setminus X_k)=j-i$ for $i\leq
k\leq j$.  Thus, all the monotonicity constraints hold with equality.

For the sub-modularity constraints, it is useful to note that for any $D_{i,j}$ with $i<j$ we have
\begin{equation*}
H(X_k|M_{k,l})=\begin{cases}
1,&k=i\text{ and }k<l\leq j,\\
1,&i<k\leq j\text{ and }k<l,\\
0,&\text{otherwise}.
\end{cases} 
\end{equation*}
Thus, $I(X_k:X_l|M_{k,l})=H(X_k|M_{k,l})-H(X_k|M_{k,l+1})$ is zero
unless $k=i$ and $l=j$ (in which case it is $1$). This establishes the
second statement, and hence completes the proof of Lemma~\ref{lem:5}.
\end{proof}

Note that the entropy vector of each of the $\frac{n(n+1)}{2}$
distributions belongs to a different extremal ray. We have thus shown
that for each extremal ray of $\overline{\Gamma^*_\cM}(P_n^\cC)$ there
is a distribution in $\cP_\cM(P_n^\cC)$ whose entropy vector lies on
that ray.  It follows by convexity that any vector that satisfies all
the Shannon constraints and the causal constraints of the marginal
scenario in $P_n^\cC$ is realisable in $P_n^\cC$ (at least
asymptotically). Since the same outer approximation is valid for
$\overline{\Gamma^{*}_{\cM}}(P_n^{\qQ})$ and any classical
distribution can be realised quantum mechanically, we have
$\overline{\Gamma^{*}_{\cM}}(P_n^{\cC}) \subseteq
\overline{\Gamma^{*}_{\cM}}(P_n^{\qQ}) \subseteq
\overline{\Gamma_{\cM}^*}(P_n^\cC)$ and 
therefore $\overline{\Gamma^{*}_{\cM}}(P_n^{\cC}) =
\overline{\Gamma^{*}_{\cM}}(P_n^{\qQ})$.

\section{Remarks on the Braunstein-Caves technique}\label{sec:2}
The generalization of the Braunstein-Caves technique to other causal structures is difficult,
as a restriction on the alphabet size of certain variables is needed. In the case of $P_4$,
we have such a restriction because $C_1$ and $C_3$ can be assumed to
be equal to the observed $A$ and $B$, whose alphabets can be
determined by observation.  In $P_n$, this can always be done for the
outermost nodes, hence, in $P_5^\cC$, for example, with observed nodes
$A$, $X$, $Y$, $Z$ and $B$, with $A$ and $B$ binary, there exists a
joint distribution $P'_{X_0X_1YZ_0Z_1}$ that gives the correct
marginal distributions, i.e., $P'_{X_aYZ_b}=P_{XYZ|A=a,B=b}$ for all
$a$ and $b$.  This is defined by setting
\begin{widetext}
$$P'_{X_0X_1YZ_0Z_1}(x,x',y,z,z')=\sum_{C_2C_3}P_{C_2}P_{C_3}P_{X|A=0,C_2}(x)P_{X|A=1,C_2}(x')P_{Y|C_2C_3}(y)P_{Z|B=0,C_3}(z)P_{Z|B=1,C_3}(z')\,
.$$
\end{widetext}
For $P'_{X_0X_1YZ_0Z_1}$, entropic inequalities can be derived with
the entropy vector approach applied to the causal structure
$\tilde{P}_3$ shown in Figure~\ref{fig:bcPn}. Note that the $P_5$
scenario is related to
bilocality~\cite{Branciard2010,Branciard2012}.
\begin{figure}
\centering
\includegraphics[width=0.65\columnwidth]{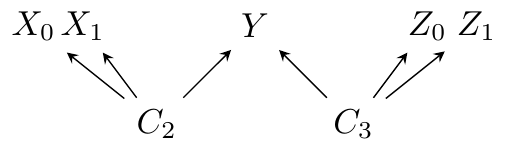}
\caption{Causal structure $\tilde{P}_3$ representing the independences of the
  variables in $P'_{X_0X_1YZ_0Z_1}$.}
\label{fig:bcPn}
\end{figure}
The Shannon and conditional independence constraints that restrict the
corresponding five variable entropy cone are marginalized to the four
triples of variables $\left\{X_0,~Y,~Z_0 \right\}$,
$\left\{X_0,~Y,~Z_1 \right\}$, $\left\{X_1,~Y,~Z_0 \right\}$ and
$\left\{X_1,~Y,~Z_1 \right\}$ and their subsets. As shown
in~\cite{Chaves2012a}, in addition to Shannon inequalities, it yields a further $36$ (in)equalities, made up of
the $7$ families listed below.  Note that the first family is a
consequence of the Shannon constraints involving all variables, and
holds independently of the causal structure.
\begin{widetext}
\begin{align}\nonumber
H(X_1|Y Z_1)+H(Z_0|X_1 Y)+H(Z_1|X_0 Y)-H(Z_0|X_0 Y) &\geq 0, \\\nonumber
H(X_0|Y Z_1)+H(Z_0|Y X_1)-H(X_0|YZ_0)+H(X_0Y)-H(X_0Z_0) &\geq 0, \\\nonumber
H(Y|X_0 Z_0)+H(X_1|Y Z_1)-H(X_1|Y Z_0) &\geq 0, \\\label{eq:others}
H(X_1|Y Z_0)+H(X_0 Y)+H(Y Z_1)-H(X_1 Y)-H(X_0 Z_1) &\geq 0, \\\nonumber
H(X_0 Y Z_1)+H(X_1 Y Z_0)-H(X_1 Y Z_1)-H(X_0 Z_0) &\geq 0, \\\nonumber
H(X_0|Y Z_1)+H(Y Z_0)-H(X_0 Z_0) &\geq 0, \\\nonumber
I(X_0 : Z_0)&=0.
\end{align}
\end{widetext}
These can be expanded to the full set by noting the symmetry between $X_0$
and $X_1$, between $Z_0$ and $Z_1$ and between $X$ and $Z$.  They form
an outer approximation to $\overline{\Gamma^*_\cM}(\tilde{P}^{\cC}_3)$.

The inequalities can be converted to the following extremal rays, with components ordered as\\
($H(X_0)$, $H(X_1)$, $H(Y)$, $H(Z_0)$, $H(Z_1)$, $H(X_0Y)$,
$H(X_0Z_0)$, $H(X_0 Z_1)$, $H(X_1Y)$, $H(X_1Z_0)$, $H(X_1Z_1)$,
$H(YZ_0)$, $H(YZ_1)$, $H(X_0YZ_0)$, $H(X_0YZ_1)$, $H(X_1YZ_0)$,
$H(X_1YZ_1)),$
\begin{alignat*}{3}
&\text{(i)} \quad &     &1 1 1 1 1  2 2 2 2 2 2 2 2  2 2 2 2  \\
&\text{(ii)} \quad &    &0 1 1 1 1  1 1 1 2 2 2 2 2  2 2 2 2  \\
&\text{(iii)} \quad &   &0 1 1 0 1  1 0 1 2 1 2 1 2  1 2 2 2  \\
&\text{(iv)} \quad &    &0 1 1 1 0  1 1 0 2 2 1 2 1  2 1 2 2  \\
&\text{(v)} \quad &     &1 0 1 1 1  2 2 2 1 1 1 2 2  2 2 2 2  \\
&\text{(vi)} \quad &    &1 0 1 0 1  2 1 2 1 0 1 1 2  2 2 1 2  \\
&\text{(vii)} \quad &   &1 0 1 1 0  2 2 1 1 1 0 2 1  2 2 2 1  \\
&\text{(viii)} \quad &  &1 1 1 0 1  2 1 2 2 1 2 1 2  2 2 2 2  \\
&\text{(ix)} \quad &    &1 1 1 1 0  2 2 1 2 2 1 2 1  2 2 2 2  \\
&\text{(x)} \quad &     &0 0 0 0 1  0 0 1 0 0 1 0 1  0 1 0 1 \\
&\text{(xi)} \quad &    &0 0 0 1 0  0 1 0 0 1 0 1 0  1 0 1 0 \\
&\text{(xii)} \quad &   &0 0 1 0 0  1 0 0 1 0 0 1 1  1 1 1 1\\
&\text{(xiii)} \quad &  &0 1 0 0 0  0 0 0 1 1 1 0 0  0 0 1 1 \\
&\text{(xiv)} \quad &   &1 0 0 0 0  1 1 1 0 0 0 0 0  1 1 0 0 \\
&\text{(xv)} \quad &    &0 0 1 1 1  1 1 1 1 1 1 1 1  1 1 1 1 \\
&\text{(xvi)} \quad &   &0 0 1 1 0  1 1 0 1 1 0 1 1  1 1 1 1 \\
&\text{(xvii)} \quad &  &0 0 1 0 1  1 0 1 1 0 1 1 1  1 1 1 1 \\
&\text{(xviii)} \quad & &1 1 1 0 0  1 1 1 1 1 1 1 1  1 1 1 1 \\
&\text{(xix)} \quad &   &1 0 1 0 0  1 1 1 1 0 0 1 1  1 1 1 1 \\
&\text{(xx)} \quad &    &0 1 1 0 0  1 0 0 1 1 1 1 1  1 1 1 1.
\end{alignat*}
These rays can be generated from those of the entropic cone of
$P_3$. To do so, let $X$, $Y$ and $Z$ be distributed according to one
of the six distributions reproducing the extremal rays of the entropy
cone of $P_3$. In the cases where $X$ is a random bit, let either
$X_0=X$ and $X_1=1$, or $X_0=1$ and $X_1=X$, or $X_0=X_1=X$, and the
same for $Z$. Doing this for all extremal rays of $P_3$, the above
extremal rays $(i)$--$(xx)$ are recovered (as well as some additional
redundant ones). This shows that the above outer approximation to
$\overline{\Gamma^*_\cM}(\tilde{P}^{\cC}_3)$ is tight: all entropy vectors that
satisfy the Shannon constraints and~\eqref{eq:others} are in
$\overline{\Gamma^*_\cM}(\tilde{P}^{\cC}_3)$.

The same technique can be applied to $P_n$ via causal structures
$\tilde{P}^{\cC}_{n-2}$. For $P_6$ this gives a total of $16$ entropic
equalities, expressing independences among the involved variables and
$153$ inequalities (including Shannon inequalities).  
In the case of $P_6$, the extremal
rays can also be generated starting from those of $P_4$ and splitting analogously to the treatment for $P_3$ above. This yields a complete characterization of $\overline{\Gamma^*_\cM}(\tilde{P}^{\cC}_4)$.

All entropic inequalities characterising
$\overline{\Gamma^*_\cM}(\tilde{P}^{\cC}_3)$ and
$\overline{\Gamma^*_\cM}(\tilde{P}^{\cC}_4)$ can be calculated without
considering the unobserved nodes: only the Shannon inequalities and
the independences among the observed variables are needed for their
derivation. Note that the same independence constraints also hold in
the analogous quantum casual structure. However, in the quantum case
the observed $X_0$ and $X_1$ as well as $Z_0$ and $Z_1$ do not coexist
and thus do not necessarily allow for a joint distribution. It is thus
not justified to analyse the causal structure $P_5^{\qQ}$ using the
related structure $\tilde{P}^{\qQ}_3$, and some of the Shannon
constraints among the variables $\left\{X_0,~X_1,~Y,~Z_0,~Z_1
\right\}$ may not hold in the quantum case. This treatment does not
therefore imply that there are no quantum violations to the classical
entropic inequalities in this approach, and, at present, we do not
know whether or not violations exist. If we allow post-quantum
non-signalling systems, however, such violations have been
found~\cite{Chaves2016}.

\end{document}